\documentclass[11pt]{article}
\usepackage{amssymb,amsthm,amsmath}
\usepackage{color}
\usepackage[usenames,dvipsnames,svgnames,table]{xcolor}
\usepackage[colorlinks=true, linkcolor=red, urlcolor=blue, citecolor=gray]{hyperref}
\usepackage{graphicx}
\usepackage{float}
\usepackage{caption}
\usepackage{subcaption}
\usepackage[margin=1in,footskip=0.25in]{geometry}
\usepackage[bottom]{footmisc}
\usepackage{booktabs}
\usepackage{enumitem}
\usepackage{algorithm}
\usepackage{algpseudocode}

\newcommand{\spara}[1]{\smallskip\noindent{\bf #1}}

\newcommand{\eqdef}{\mathbin{\stackrel{\rm def}{=}}}
\makeatletter
\def\hlinewd#1{%
	\noalign{\ifnum0=`}\fi\hrule \@height #1 \futurelet
	\reserved@a\@xhline}
\makeatother

\newtheorem{theorem}{Theorem}

\newtheorem{corollary}[theorem]{Corollary}
\newtheorem{lemma}[theorem]{Lemma}

\newtheorem{claim}[theorem]{Claim}

\newtheorem{definition}{Definition}
\newtheorem{problem}[definition]{Problem}

\newtheorem*{rep@theorem}{\rep@title}
\newcommand{\newreptheorem}[2]{%
	\newenvironment{rep#1}[1]{%
		\def\rep@title{#2 \ref{##1}}%
		\begin{rep@theorem}}%
		{\end{rep@theorem}}}
\makeatother
\newreptheorem{theorem}{Theorem}
\newreptheorem{claim}{Claim}

\newcommand{\R}{\mathbb{R}}

\newcommand{\norm}[1]{\|#1\|}
\newcommand{\Ak}{A_{\setminus k}}

\newcommand{\Ar}{A_{\setminus r}}
\newcommand{\Ap}{A_{\setminus p}}

\DeclareMathOperator*{\argmin}{arg\,min}

\DeclareMathOperator{\tr}{tr}
\DeclareMathOperator{\rank}{rank}

  \usepackage{nth}
  \usepackage{intcalc}

\title{Projection-Cost-Preserving Sketches: \\ Proof Strategies and  Constructions}

\author{
	Cameron Musco\\  UMass Amherst\\  \texttt{cmusco@cs.umass.edu}
	\and
	Christopher Musco\\  New York University\\  \texttt{cmusco@nyu.edu}
}

\begin{document}

\maketitle
\begin{abstract}
In this note we illustrate how common matrix approximation methods, such as random projection and random sampling, yield projection-cost-preserving sketches, as introduced in \cite{feldman2013turning,cohen2015dimensionality}. A projection-cost-preserving sketch is a matrix approximation which, for a given parameter $k$, approximately preserves the distance of the target matrix to all $k$-dimensional subspaces. Such sketches have applications to scalable algorithms for linear algebra, data science, and machine learning.
Our goal is to simplify the presentation of proof techniques introduced in \cite{cohen2015dimensionality} and \cite{cohen2017input} so that they can serve as a guide for future work. We also refer the reader to \cite{chowdhury2019structural}, which gives a similar simplified exposition of the proof covered in Section \ref{sec:matrixApprox}.
\end{abstract}

\section{Projection-Cost-Preserving Sketches}

A projection-cost-preserving sketch is a matrix compression that preserves the distance of a matrix's columns
to any $k$-dimensional subspace. Let $\norm{M}_F^2  = \sum_{i,j}M_{i,j}^2$ denote the squared Frobenius norm of a matrix $M$. Formally we define:

\begin{definition}[Projection-Cost-Preserving Sketch]\label{def:pcp}
$\tilde A \in \R^{n \times m}$ is an $(\epsilon,c,k)$-projection-cost-preserving sketch of $A \in \R^{n \times d}$ if, for any orthogonal projection matrix $P \in \R^{n \times n}$ with rank at most $k$, 
\begin{align}\label{eq:pcp}
(1-\epsilon) \norm{A-PA}_F^2  \le \norm{{\tilde A}-{P}{\tilde A}}_F^2 + c \le (1+\epsilon) \norm{{A}-{PA}}_F^2. 
\end{align}
Here $c$ is a constant that is independent of ${P}$ (but may depend on ${A},{\tilde A},\epsilon,k$).
\end{definition}

In typical applications, $m \ll d$, so $\tilde{A}$ has fewer columns that $A$. It can serve as a surrogate in solving a number of low-rank optimization problems, such as PCA or $k$-means clustering, in which the goal is to chose a $k$-dimensional subspace from some set that is as close as possible to the input matrix.  When $m \ll d$, using $\tilde{A}$ in place of $A$ can lead to significant computational savings in terms of runtime, memory, and communication cost. 

\subsection{Constrained Low-Rank Approximation}
For example, a projection-cost-preserving sketch can be used to approximately solve any  problem of the form:
\begin{problem}[Constrained Low-Rank Approximation]\label{prob:clr}
Let $\mathcal{S}_k$ be the set of all orthogonal projection matrices  in $\R^{n \times n}$ with rank $\le k$. Let $\mathcal{T}$ be any subset of $\mathcal{S}_k$. The constrained low-rank approximation problem over set $\mathcal{T}$ is to find:
$P^\star \in \argmin_{P \in \mathcal T} \norm{A - PA}_F^2.
$
\end{problem}
\noindent A simple manipulation of the bound of Definition \ref{def:pcp} yields:
\begin{claim}
If  $\tilde{A}$ is an $(\epsilon,c,k)$-PCP for $A$, then for any $\mathcal{T} \subseteq \mathcal{S}_k$, if $\tilde{P} \leq \gamma \cdot \min_{P \in \mathcal T} \norm{ \tilde{A} - P\tilde{A} }_F^2$ for some $\gamma \ge 1$, then:
$$
\norm{ A - \tilde{P} A }_F^2 \leq \frac{(1+\epsilon) \gamma}{1-\epsilon} \cdot \min_{P \in \mathcal T} \norm{A-P A}_F^2 + \frac{(1-\gamma)c}{1-\epsilon}.
$$
\end{claim}
Note that if $c$ is positive (as will typically be the case), $\frac{(1-\gamma)c}{1-\epsilon} < 0$, and thus $\tilde P$ gives a relative error approximation to the optimum. This is also true if $c$ is negative and $\gamma \le 1 + \frac{\min_{P \in \mathcal T} \norm{A-P A}_F^2}{|c|}$.


Two important cases of Problem \ref{prob:clr} are vanilla low-rank approximation, when $\mathcal T = \mathcal S_k$ and $k$-means clustering, when $\mathcal{T}$ is the set of projections corresponding to the set of \emph{cluster indicator matrices}. See \cite{cohen2015dimensionality} for details.

\subsection{Sketch Constructions}
It has been shown that a wide variety of dimensionality reduction methods can be used to obtain projection-cost-preserving sketches with positive $c \geq 0$ and dimension $m = \tilde O(k/\epsilon^q)$ for $q \in \{1,2\}$. See Table \ref{tab:constructions} below.

In Section \ref{sec:matrixApprox} we show how to prove that a dimensionality reduction method yields a projection-cost-preserving sketch by appealing to the well-studied matrix approximation guarantees of \emph{subspace embedding}, \emph{approximate matrix multiplication}, and \emph{Frobenius norm preservation}. This proof mirrors the more general proof of  \cite{cohen2015dimensionality} and the proof presented in \cite{chowdhury2019structural}.
In Section \ref{sec:spectralApprox} we show how to prove the projection-cost-preserving sketch guarantee in an alternative way: starting from a spectral approximation bound of the form $(1-\epsilon) A A^T - \lambda I \preceq \tilde A \tilde A^T \preceq (1+\epsilon) A A^T + \lambda I$, where $M \preceq N$ denotes that $N - M$ is positive semidefinite, $I$ is the $n \times n$ identity matrix, and $\lambda$ is an appropriately chosen regularization parameter. This proof mirrors that in \cite{cohen2017input}. 

The two proofs are closely related and both follow a strategy of decomposing $A$ into the projections onto the singular vectors corresponding to its large (head) and small (tail) singular values. Error terms corresponding to these components are then bounded using well-studied matrix approximation guarantees (Section \ref{sec:matrixApprox}) or using the above spectral approximation bound  (Section \ref{sec:spectralApprox}). \cite{chowdhury2019structural} further discusses how the two proof strategies can be unified under a general approach.

\begin{table}[h]
\begin{center}
\begin{tabular}{c | c | c}
Method & Dimension $m$ &  Reference\\
\hline
 SVD & $\lceil k/\epsilon \rceil$ &Theorem 7 of  \cite{cohen2015dimensionality} \\ 
 Approximate SVD & $\lceil k/\epsilon \rceil$ & Theorems 8,9 of \cite{cohen2015dimensionality} \\  
 Random Projection & $O(k/\epsilon^2)$ & Theorem 12 of \cite{cohen2015dimensionality}    \\
 Non-Oblivious Random Projection\footnotemark & $O(k/\epsilon)$ & Theorem 16 of \cite{cohen2015dimensionality}    \\
Ridge Leverage Score Column Sampling & $O(k\log k/\epsilon^2)$ & Theorem 6 of \cite{cohen2017input}    \\
Leverage Score + Residual Column Sampling & $O(k\log k/\epsilon^2)$ & Theorem 14 of \cite{cohen2015dimensionality}    \\
Deterministic Column Selection & $O(k/\epsilon^2)$ & Theorem 15 of \cite{cohen2015dimensionality}    \\
Frequent Directions Sketch & $\lceil k/\epsilon \rceil + k $ & Theorem 31 of \cite{masters}    \\
\end{tabular}
\end{center}
\caption{Known projection-cost-preserving sketch constructions. All theorem references are to the arXiv versions of the cited papers. For randomized constructions, dependencies on success probability are hidden.}
\label{tab:constructions}
\end{table}
\footnotetext{In this method, compute $Z \in \R^{d \times m}$ with orthonormal columns spanning  the rows of $\Pi A$ where $\Pi \in \R^{m \times n}$ is a random projection matrix. Then let $\tilde A = AZ$.}

\section{Proof Via Matrix Approximation Primitives}\label{sec:matrixApprox}
We start by  defining three well-studied matrix approximation primitives:
\begin{definition}[Subspace Embedding]\label{def:subspace}
$S \in \R^{d \times m}$ is an $\epsilon$-subspace embedding for $M \in \R^{n \times d}$ if $\forall x \in \R^n$, $\left | \norm{ x^T M }_2^2 - \norm{ x^T MS }_2^2 \right | \le \epsilon \norm{ x^T M}_2^2$.
\end{definition}

\begin{definition}[Approximate Matrix Multiplication]\label{def:amm}
$S \in \R^{d \times m}$ satisfies $\epsilon$-approximate matrix multiplication for $M \in \R^{n \times d}$, $N \in \R^{d \times p}$ if $\norm{MN^T - MSS^T N}_F^2 \le \epsilon \cdot \norm{M}_F \cdot \norm{N}_F$.
\end{definition}

\begin{definition}[Frobenius Norm Preservation]\label{def:frob}
$S \in \R^{d \times m}$ satisfies $\epsilon$-Frobenius norm preservation for $M \in \R^{n \times d}$ if $\left | \norm{M}_F^2 - \norm{MS}_F^2 \right | \le \epsilon \norm{M}_F^2$.
\end{definition}

We will also define a useful notion of splitting any matrix into the part in the span of its top $r$ singular vectors and the part outside this span:
\begin{definition}[Head-Tail Split]\label{def:headTail}
For any $M \in \R^{n \times d}$ consider the singular value decomposition $U\Sigma V^T = M$, where $U \in \R^{n \times \rank(M)}, V \in \R^{d \times \rank(M)}$ have orthonormal columns (the left and right singular vectors of $M$ respectively), and $\Sigma$ is a nonnegative diagonal matrix with entries equal to $M$'s singular values $\sigma_{1} \ge \sigma_{2} \ge \ldots \sigma_{\rank(M)} > 0$. 

For any $r \le \rank(M)$ let $U_r \in \R^{n \times r}$, $V_r \in \R^{d \times r}$ denote the first $r$ columns of $U,V$ respectively and let $M_r = U_r U_r^T M = M V_r V_r^T$ and $M_{\setminus r} = M - M_r$. Note that $M_r$ is the optimal $r$-rank approximation of $M$: $M_r = \argmin_{\rank-r\, C} \norm{M-C}_F^2$.
\end{definition} 

We now have the following theorem, which is very similar to Theorem 2 of \cite{chowdhury2019structural}:
\begin{theorem}[Projection-cost-preserving sketch via matrix approximation]\label{thm:matrixApprox}
If $S \in \R^{d \times m}$:
\begin{enumerate}
\item Is an $\frac{\epsilon}{3}$-subspace embedding (Definition \ref{def:subspace}) for $A_k$.
\item Satisfies $\frac{\epsilon}{6\sqrt{k}}$-approximate matrix multiplication (Definition \ref{def:amm}) for $\Ak$, $\Ak$.
\item Satisfies $\frac{\epsilon}{6\sqrt{k}}$-approximate matrix multiplication (Definition \ref{def:amm}) for $\Ak$, $V_k$.
\item Satisfies $\frac{\epsilon}{6}$-Frobenius norm preservation (Definition \ref{def:frob}) for $\Ak$.
\end{enumerate}
Then $\tilde A = AS$ is an $(\epsilon,0,k)$-projection-cost-preserving sketch of $A$.
\end{theorem}

\begin{proof}
For any orthogonal projection matrix $P \in \R^{n \times n}$ with rank at most $k$, let $Y = I - P$, where $I$ is the $n \times n$ identity matrix. To prove the theorem, it suffices to show that:
\begin{align}\label{eq:initial}
\left | \norm{YA}_F^2 - \norm{YAS}_F^2 \right | \leq \epsilon\norm{YA}_F^2.
\end{align}
This immediately gives \eqref{eq:pcp} with $c = 0$. To prove \eqref{eq:initial} we will decompose the error into head and tail, terms following Definition \ref{def:headTail}. We write $A = A_k + \Ak$ and then rewrite \eqref{eq:initial} as:
\begin{align}\label{eq:initial2}
\left | \norm{Y (A_k + \Ak) }_F^2 - \norm{ Y (A_k + \Ak)S }_F^2 \right | \le \epsilon \norm{YA}_F^2.
\end{align}
Expanding out the left hand side, using that $\norm{M}_F^2 = \tr(MM^T)$ and that $\tr(Y A_k \Ak^T Y) = 0$ since $A_k$ and $\Ak$ have orthogonal row spans we can see that to show \eqref{eq:initial2} it suffices to show:
\begin{align}
\underbrace{\left | \tr(Y A_k A_k^T Y)- \tr(Y A_k SS^T A_k^T Y) \right  |}_{\text{head term}}& + \underbrace{\left | \tr(Y \Ak \Ak^T Y) - \tr(Y\Ak SS^T \Ak^T Y) \right |}_{\text{tail term}}\nonumber\\
&+ \underbrace{2 \left |\tr(YA_k SS^T \Ak^T Y) \right |}_{\text{cross term}} \le \epsilon \norm{YA}_F^2.\label{eq:initial3}
\end{align}
We now bound the three terms of \eqref{eq:initial3} separately.


\begin{claim}[Head Bound - via subspace embedding]\label{clm:head}
Under the assumptions of Theorem \ref{thm:matrixApprox}, for any orthogonal projection matrix $P \in \R^{n \times n}$ with rank at most $k$ and $Y = I - P$:
\begin{align}\label{eq:head}
| \tr(Y A_k A_k^T Y) - \tr(Y A_k SS^T A_k^T Y) | \le \frac{\epsilon}{3} \norm{YA}_F^2.
\end{align}
\end{claim}
\begin{proof}
By the assumption that $S$ is an $\frac{\epsilon}{3}$-subspace embedding for $A_k$ (Definition \ref{def:subspace}) we have:
\begin{align*}
| \tr(Y A_k A_k^T Y) - \tr(Y A_k SS^T A_k^T Y) | &= \left |\norm{YA_k}_F^2 - \norm{YA_kS}_F^2 \right  |\nonumber\\
&\le \frac{\epsilon}{3} \norm{YA_k}_F^2 \le \frac{\epsilon}{3} \norm{YA}_F^2.
\end{align*}
The second to last inequality follows from the subspace embedding property. In particular, let $(YA_k)_i$ and $(YA_kS)_i$ denote the $i^{th}$ rows of $YA_k$ and $YA_kS$ respectively. Then by the subspace embedding property we have $\left |\norm{(YA_k)_i}_2^2 - \norm{(YA_kS)_i}_2^2 \right | \le \frac{\epsilon}{3} \norm{(YA_k)_i}_2^2$.
\end{proof}


\begin{claim}[Tail Bound - via approximate matrix multiplication \& Frobenius norm preservation]\label{clm:tail}
Under the assumptions of Theorem \ref{thm:matrixApprox}, for any orthogonal projection matrix $P \in \R^{n \times n}$ with rank at most $k$ and $Y = I - P$:
\begin{align}\label{eq:tail}
| \tr(Y \Ak \Ak^T Y) - \tr(Y \Ak SS^T \Ak^T Y) |  \le \frac{\epsilon}{3} \norm{YA}_F^2.
\end{align}
\end{claim}
\begin{proof}
We rewrite the tail term as:
\begin{align}
| \tr(Y \Ak \Ak^T Y) - \tr(Y \Ak SS^T \Ak^T Y) | &= \left |\norm{(Y \Ak}_F^2 - \norm{Y \Ak S}_F^2 \right |\nonumber\\
&= \left |\norm{(I-P) \Ak}_F^2 - \norm{(I-P) \Ak S}_F^2 \right |\nonumber\\
& = \left |(\norm{ \Ak}_F^2 - \norm{P \Ak}_F^2)  - (\norm{\Ak S}_F^2 - \norm{P\Ak S}_F^2) \right |\nonumber\\
& \le \left | \norm{\Ak}_F^2 - \norm{\Ak S}_F^2 \right | + \left | \norm{P \Ak}_F^2 - \norm{P \Ak S}_F^2\right |.\label{eq:tailSplit}
\end{align}
The third line follows from the Pythagorean theorem. By the assumption that $S$ satisfies $\frac{\epsilon}{6}$-Frobenius norm preservation (Definition \ref{def:frob}) for $\Ak$ we can bound the first term in \eqref{eq:tailSplit} by:
\begin{align}\label{eq:tail1}
\left | \norm{\Ak}_F^2 - \norm{\Ak S}_F^2 \right |  \le \frac{\epsilon}{6} \norm{\Ak}_F^2 \le \frac{\epsilon}{6} \norm{YA}_F^2,
\end{align}
where the last inequality follows from the fact that $A_k = \argmin_{rank-k\, C} \norm{A-C}_F^2$. Thus $\norm{\Ak}_F^2 = \norm{A-A_k}_F^2 \le \norm{A-PA}_F^2$ for any rank-$k$ projection $P$. We write the second term of \eqref{eq:tailSplit} as:
\begin{align*}
| \norm{ P \Ak }_F^2 - \norm{ P \Ak S }_F^2 | = | \tr(P[\Ak \Ak^T - \Ak SS^T \Ak^T]P)|.
\end{align*}
$P[\Ak \Ak^T - \Ak SS^T \Ak^T]P$ has rank at most $k$ since $P$ has rank at most $k$. Letting $\lambda_1,\ldots, \lambda_k$ denote its eigenvalues we have:
\begin{align}
| \norm{ P \Ak }_F^2 - \norm{ P \Ak S }_F^2 | &= | \tr(P[\Ak \Ak^T - \Ak SS^T \Ak^T]P)| \nonumber\\
&=  \left |\sum_{i=1}^k \lambda_i \right | \le \sum_{i=1}^k \left |\lambda_i \right | \le \sqrt{k} \cdot \left (\sum_{i=1}^k \lambda_i^2 \right ) \nonumber\\
 &= \sqrt{k} \cdot \norm{P[\Ak \Ak^T - \Ak SS^T \Ak^T]P}_F^2 \nonumber\\
&\le  \sqrt{k} \cdot \norm{\Ak \Ak^T - \Ak SS^T \Ak^T}_F^2.\label{eq:traceBound}
\end{align}
The last inequality follows since $P$ is a projection matrix and can only decrease the Frobenius norm. By our assumption that $S$ satisfies $\frac{\epsilon}{6\sqrt{k}}$-approximate matrix multiplication (Definition \ref{def:amm}) for $\Ak,\Ak$ we thus can bound:
\begin{align}\label{eq:tail2}
| \Vert P \Ak \Vert_F^2 - \Vert P \Ak S \Vert_F^2 | \le \frac{\epsilon}{6} \norm{\Ak}_F^2 \le \frac{\epsilon}{6} \norm{YA}_F^2,
\end{align}
where the second inequality again follows since $\Ak = A-A_k$ is the error of the best rank-$k$ approximation to $A$. Plugging \eqref{eq:tail1} and \eqref{eq:tail2} back into \eqref{eq:tailSplit} we have:
\begin{align*}
| \tr(Y \Ak \Ak^T Y) - \tr(Y \Ak SS^T \Ak^T Y) |  \le \frac{\epsilon}{3} \norm{YA}_F^2,
\end{align*}
which gives \eqref{eq:tail} and completes the claim.
\end{proof}


\begin{claim}[Cross Term Bound - via approximate matrix multiplication]\label{clm:cross}
Under the assumptions of Theorem \ref{thm:matrixApprox}, for any orthogonal projection matrix $P \in \R^{n \times n}$ with rank at most $k$ and $Y = I - P$:
\begin{align}\label{eq:cross}
2 \left |\tr(YA_k SS^T \Ak^T Y) \right | \le \frac{\epsilon}{3} \norm{YA}_F^2.
\end{align}
\end{claim}
\begin{proof}
Let $C = AA^T$ and let $C^+$ be its pseudoinverse. Writing $A$ in its SVD, $A= U\Sigma V^T$, we have $C^+ = U\Sigma^{-2} U^T$. We let $C^{+/2} = U \Sigma^{-1} U^T$. We can bound the cross term as:
\begin{align}
2 \left |\tr(YA_k SS^T \Ak^T Y) \right | &= 2| \tr(Y CC^+ A_k SS^T \Ak^T Y)| \tag{Since the columns of $A_k$ fall within the column span of $C$}\nonumber\\
&= 2|\tr(Y^2 CC^+ A_k SS^T \Ak^T)|\tag{By the cyclic property  of the trace}\nonumber \\
&=2 |\tr(Y CC^+ A_k SS^T \Ak^T)| \tag{Since $Y = I-P$ is an orthogonal projection so $Y^2 = Y$.}\nonumber\\
&=2 |\tr((Y CC^{+/2})(C^{+/2} A_k SS^T \Ak^T)| \nonumber\\
&\leq 2\sqrt{\tr(YCC^{+/2} C^{+/2}C Y)} \cdot \sqrt{\tr(\Ak SS^T A_k^T C^{+/2}C^{+/2} A_k SS^T \Ak^T)}.\label{eq:crossSplit}
\end{align}
The last inequality follows from Cauchy-Schwarz. The first term of \eqref{eq:crossSplit} can be bounded by:
\begin{align}
\sqrt{\tr(YCC^{+/2} C^{+/2}C Y)} = \sqrt{\tr(YCY)} = \sqrt{\tr(YAA^TY)} = \norm{YA}_F.\label{eq:cross1}
\end{align}
We bound the second term of \eqref{eq:crossSplit} by using the SVD to write $A_k = U_k \Sigma_k V_k^T$, where $U_k,V_k$ are as in Definition \ref{def:headTail} and $\Sigma_k \in \R^{k \times k}$ is the top left $k \times k$ submatrix of $\Sigma$. We have:
\begin{align}\label{eq:secondTermNorm}
\sqrt{\tr(\Ak SS^T A_k^T C^{+/2}C^{+/2} A_k SS^T \Ak^T)}
&= \sqrt{\tr(\Ak SS^T V_k \Sigma_k U_k^T U \Sigma^{-2} U^T U_k \Sigma_k V_k^T SS^T \Ak^T)} \nonumber\\
&= \sqrt{\tr(\Ak SS^T V_k V_k^T SS^T \Ak^T)}\nonumber\\
&= \Vert \Ak SS^T V_k \Vert_F = \Vert \Ak SS^T V_k - \Ak V_k \Vert_F.
\end{align}
The last line follows from the fact that the rows of $\Ak$ are orthogonal to the columns of $V_k$ and thus $\Ak V_k = 0$.
By the assumption that $S$ satisfies $\frac{\epsilon}{6\sqrt{k}}$-approximate matrix multiplication for $\Ak,V_k$, we thus have
\begin{align}
\sqrt{\tr(\Ak SS^T A_k^T C^{+/2}C^{+/2} A_k SS^T \Ak^T)} &\leq \frac{\epsilon}{4\sqrt{k}} \norm{\Ak}_F \cdot \norm{V_k}_F\nonumber\\
&\le \frac{\epsilon}{6\sqrt{k}} \norm{YA}_F \cdot \sqrt{k} = \frac{\epsilon}{6}\norm{YA}_F.\label{eq:cross2}
\end{align}
Plugging \eqref{eq:cross1} and \eqref{eq:cross2} back into \eqref{eq:crossSplit} we have:
\begin{align*}
2 \left |\tr(YA_k SS^T \Ak^T Y) \right | \le \frac{\epsilon}{3} \norm{YA}_F^2,
\end{align*}
which gives \eqref{eq:cross} and completes the claim.
\end{proof}

\medskip
\spara{Completing the Proof:}
\medskip

\noindent Finally, we combine the head, tail and cross term bounds of Claims \ref{clm:head}, \ref{clm:tail}, and \ref{clm:cross} to give:
\begin{align*}
{\left | \tr(Y A_k A_k^T Y)- \tr(Y A_k SS^T A_k^T Y) \right  |}& +{\left | \tr(Y \Ak \Ak^T Y) - \tr(Y\Ak SS^T \Ak^T Y) \right |}+{2 \left |\tr(YA_k SS^T \Ak^T Y) \right |}\\
&\le \frac{\epsilon}{3} \norm{YA}_F^2 + \frac{\epsilon}{3} \norm{YA}_F^2 + \frac{\epsilon}{3} \norm{YA}_F^2 = \epsilon \norm{YA}_F^2.
\end{align*}
This yields \eqref{eq:initial3}, which completes the proof of Theorem \ref{thm:matrixApprox}.
\end{proof}

\subsection{Constructions Satisfying Theorem \ref{thm:matrixApprox}}
Theorem \ref{thm:matrixApprox} can be used to prove that a number of constructions of $S$ give projection-cost-preserving sketches. A simple example is when $S$ is a random projection matrix. In fact, any projection matrix satisfying a certain \emph{Johnson-Lindenstrauss moment property} suffices.

\begin{definition}[$(\epsilon, \delta, \ell)$-JL moment property, \cite{KaneNelson:2014}]
	A matrix $S \in \R^{d\times m}$ satisfies the $(\epsilon, \delta, \ell)$-JL moment property if for any $x \in \R^d$ with $\|x\|_2 = 1$,
	\begin{align*}
		\mathbb{E}_{S} |\|x^T S\|_2^2 - 1|^\ell \leq \epsilon^\ell\cdot \delta.
	\end{align*}
\end{definition}

\begin{lemma}[Projection-cost-preservation from JL moment property]\label{lem:jlmoment}
	If $S$ satisfies the $(\frac{\epsilon}{6\sqrt{k}}, \delta, \ell)$-JL moment property and the $(\frac{\epsilon}{3}, \frac{\delta}{9^k}, \ell)$-JL moment property for any $\ell \geq 2$, then with probability $\geq 1-4\delta$, $\tilde A = AS$ is an $(\epsilon,0,k)$-projection-cost-preserving sketch of $A$.
\end{lemma}
\begin{proof}
	It is well known that if $S$ satisfies the $(\frac{\epsilon}{3}, \frac{\delta}{9^k}, \ell)$-JL moment property for any $\ell >0$, then with probability $\ge 1-\delta$, $S$ is an $\frac{\epsilon}{3}$-subspace embedding for $A_k$ since $A_k$ has rank $k$. The proof follows from a net argument, as given in \cite{Sarlos:2006} or Theorem 2.1 of \cite{woodruff2014sketching}. 
	
	We also have from Theorem 2.8 in \cite{woodruff2014sketching} that if  $S$ satisfies the $(\frac{\epsilon}{6\sqrt{k}}, \delta, \ell)$-JL moment property for any $\ell \geq 2$, then  $S$ satisfies the $\frac{\epsilon}{6\sqrt{k}}$-approximate matrix multiplication property with probability $\ge 1-\delta$ for \emph{any} pair of matrices. 
	
	Finally, we claim that if $S$ satisfies the $(\frac{\epsilon}{6}, \delta, \ell)$-JL moment property for any $\ell > 0$, then $S$ satisfies the $\frac{\epsilon}{6}$-Frobenius norm preservation condition for \emph{any} matrix $M \in \R^{n\times d}$, with probability $1 - \delta$. In particular, let $m_1,\ldots, m_n$ denote the rows of $M$. Let $\tilde{\epsilon}$ denote $\epsilon/6$. We have:
	\begin{align*}
	\Pr\left[\left|\|MS\|_F^2 - \|M\|_F^2\right| > \tilde{\epsilon} \|M\|_F^2 \right] &\leq \tilde{\epsilon}^{-\ell}\|M\|_F^{-2\ell} \cdot \mathbb{E}\left[\left|\|MS\|_F^2 - \|M\|_F^2\right|^\ell\right] \\
	&= \tilde{\epsilon}^{-\ell}\|M\|_F^{-2\ell} \cdot \mathbb{E}\left[\left|\sum_{i=1}^n \|m_i^TS\|_2^2 - \|m_i\|_2^2\right|^\ell\right] \\
	&\leq \tilde{\epsilon}^{-\ell}\|M\|_F^{-2\ell} \cdot \left[\sum_{i=1}^n \left[\mathbb{E}\left|\|m_i^TS\|_2^2 - \|m_i\|_2^2\right|^\ell\right]^{1/\ell} \right]^\ell
	\end{align*}
	The last inequality follows from Minkowski's inequality. Then we use the JL-moment property:
	
	\begin{align*}
	\tilde{\epsilon}^{-\ell}\|M\|_F^{-2\ell}\left[\sum_{i=1}^n \left[\mathbb{E}\left|\| m_i^TS\|_2^2 - \|m_i\|_2^2\right|^\ell\right]^{1/\ell} \right]^\ell 
	&\leq \tilde{\epsilon}^{-\ell}\|M\|_F^{-2\ell}\cdot \left[\sum_{i=1}^n \left[\tilde{\epsilon}^\ell\cdot\delta\cdot\|m_i\|_2^{2\ell}\right]^{1/\ell} \right]^\ell \\ 
	&\leq \tilde{\epsilon}^{-\ell}\|M\|_F^{-2\ell}\cdot \left[\tilde{\epsilon}\cdot \delta^{1/\ell}\cdot \|M\|_F^2 \right]^\ell
	\leq \delta.
	\end{align*}
	
	Applying a union bound, we have that all four condition of Theorem \ref{thm:matrixApprox} hold with probability $\geq 1- 4\delta$, which completes the proof of Lemma \ref{lem:jlmoment}.
	\end{proof}
	
Many standard random projection matrices, including classic random Gaussian matrices, random Rademacher matrices, and fast and sparse JL transforms can be shown to satisfy the requirements of Lemma \ref{lem:jlmoment} with varying embedding dimensions. For example, we have the following
\begin{corollary}[Dense Random Projection]\label{cor:subGaussian}
Let $S \in \R^{d \times m}$ be a matrix with each entry set independently to $S_{i,j} = \frac{\mathcal{N}(0,1)}{\sqrt{m}}$ where $\mathcal{N}(0,1)$ is a standard normal random variable. If $m \ge c \cdot \frac{k + \log(1/\delta)}{\epsilon^2}$ for a sufficiently large universal constant $c$ then with probability $\ge 1-\delta$, $\tilde A = AS$ is an $(\epsilon,0,k)$-projection-cost-preserving sketch of $A$.
\end{corollary}
\begin{proof}
	When $S$ is a random Gaussian matrix, $\|x^TS\|_2^2$ is sum of independent Chi-squared random variables. We can thus directly apply a moment bound for sub-exponential random variables \cite{Wainwright:2019,Vershynin:2018} to establish that when $m \ge c \cdot \frac{k + \log(1/\delta)}{\epsilon^2}$, $S$ satisfies both the  $(\frac{\epsilon}{6\sqrt{k}}, \frac{\delta}{4}, \ell)$-JL moment property and the $(\frac{\epsilon}{3}, \frac{\delta/4}{9^k}, \ell)$-JL moment property. Applying Lemma \ref{lem:jlmoment} completes the proof. 
\end{proof}
For information on other random projection matrices that can be analyzed using Lemma \ref{lem:jlmoment}, see \cite{cohen2015dimensionality} and \cite{chowdhury2019structural}. Some of these matrices can be applied faster or stored in less space than the dense random projection of Corollary \ref{cor:subGaussian} because they are sparse or structured. 

Finally, we note that the conditions of Theorem \ref{thm:matrixApprox} can also be satisfied by a simple sampling scheme, which was proposed in \cite{cohen2015dimensionality} and also analyzed in \cite{chowdhury2019structural}:
\begin{corollary}[Leverage Score + Residual Sampling]\label{cor:sampling} Consider $A$ with SVD $A = U\Sigma V^T$.
For every  $i \in 1,\ldots n$ let $p_i = \frac{\norm{(U_k)_i}_2^2}{2k} + \frac{\norm{(A-A_k)_i}_2^2}{2\norm{A-A_k}_F^2}$. Let $S$ be a sampling matrix selecting $m$ columns of $A$ where each column of $S$ is set independently to $\frac{1}{\sqrt{mp_i}} e_i$ with probability  $p_i$, where $e_i$ is the $i^{th}$ standard basis vector. Then for $m \ge \frac{c k \log(k/\delta)}{\epsilon^2}$ for some universal constant $c$, with probability $\ge 1-\delta$, $S$ satisfies all four requirements of Theorem \ref{thm:matrixApprox} and so $\tilde A = AS$ is an $(\epsilon,0,k)$-projection-cost-preserving sketch of $A$.
\end{corollary}

\subsection{Proof Variants}

We briefly mention a few variants on the proof of Theorem \ref{thm:matrixApprox} that may be useful.

\medskip
\spara{Head-Tail Split Using an Approximate Basis:}
\medskip

Corollary \ref{cor:sampling} requires sampling by the leverage scores of the rank-$k$ subspace spanned by $A$'s true top $k$ singular vectors. This is to ensure that $S$ is an $\frac{\epsilon}{3}$-subspace embedding for $A_k$ (requirement (1) of Theorem \ref{thm:matrixApprox}.) It can be shown that the leverage scores of any approximate subspace (computed e.g. via an input sparsity time sketching method) can be used instead (see Theorem 14 of \cite{cohen2015dimensionality}, arXiv version). The proof requires splitting $A$ into `approximate head and tail terms' $AZZ^T$ and $A(I-ZZ^T)$ where $Z \in \R^{d \times O(k)}$ has orthonormal columns and where $A-AZZ^T$ is a near optimal low-rank approximation of $A$. See Lemma 10 of \cite{cohen2015dimensionality}, arXiv version.

\medskip
\spara{Allowing a Constant Error Term}
\medskip

Recall that Definition \ref{def:pcp} allows $\norm{\tilde A - P \tilde A}_F^2$ to approximate $\norm{A-PA}_F^2$ up to an additive constant $c$, which can depend on $A, \tilde A,\epsilon,k$, but not on $P$. The use of such a term is useful, e.g., when $\tilde A$ is obtained by taking a low-rank approximation of A and consistently  underestimates the cost $\norm{A-PA}_F^2$ (see Theorem 7,  8, 9 of \cite{cohen2015dimensionality}, arXiv version). 

Allowing a constant term can be useful also in loosening the requires for Theorem \ref{thm:matrixApprox}. For example, requirement 4, that $S$ preserves $\norm{\Ak}_F^2$ up to $\frac{\epsilon}{6}$ error, can be relaxed: we need only require that $\norm{\Ak S}_F^2 \le c_1 \norm{\Ak}_F^2$ for some constant $c_1 \ge 1$. As long as $S$ preserves the norm in expectation, such a constant error bound is easy to show via Markov's inequality. The Frobenius norm requirement is used in proving \eqref{eq:tail1}, which bounds the first term of \eqref{eq:tailSplit}. However, note that this term $|\norm{\Ak}_F^2 - \norm{\Ak S}_F^2|$ is a constant independent of $P$ and is bounded by $c_1 \norm{\Ak}_F^2$ as long as $\norm{\Ak S}_F^2 \le c_1\norm{\Ak}_F^2$. Setting $c = - |\norm{\Ak}_F^2 - \norm{\Ak S}_F^2|$, we can absorb this term into the projection-cost-preserving sketching bound. We can also see that this is sufficient to achieve a relative error bound in Claim \ref{prob:clr} as long as $\gamma$ is a less than a constant sufficiently close to $1$.
 See Lemma 7 of \cite{musco2017sublinear}, arXiv version for an example application of this technique.

\medskip
\spara{Head-Tail Split Using a Higher Dimension:}
\medskip

In some cases, we can relax the requirements of Theorem \ref{thm:matrixApprox} or give a stronger guarantee (e.g., a projection-cost-preserving sketch where error is measured in the spectral rather than the Frobenius norm) by splitting $A = A_r + \Ar$ for $r > k$ instead of $A = A_k + \Ak$. If we set $r = ck/\epsilon + k$ we have $\norm{\Ar}_2^2 \le \frac{\epsilon}{ck} \norm{\Ak}_F^2$, which can be a useful bound. See Lemmas 7 and 8 of \cite{musco2017sublinear}, arXiv version for an example of this technique. Also see Theorem \ref{thm:spectralApprox} in Section \ref{sec:spectralApprox}.

\section{Proof via Spectral Approximation}\label{sec:spectralApprox}

We now show an alternative strategy  to proving that a sketching matrix $S \in \R^{d \times m}$ yields $\tilde A = AS$ which is a projection-cost-preserving sketch. This proof strategy was presented in \cite{cohen2017input}.
\begin{theorem}[Projection-cost-preserving sketch via spectral approximation]\label{thm:spectralApprox}
If $S \in \R^{d \times m}$:
\begin{enumerate}
\item Satisfies $\left (1-\frac{\epsilon}{24}  \right ) AA^T - \lambda I \preceq AS S^T A^T \preceq \left (1+\frac{\epsilon}{24} \right  ) A A^T + \lambda I$ for $\lambda = \frac{\epsilon \cdot \norm{A - A_k}_F^2}{24k}$.\footnote{Equivalently, for any $x \in \R^n$, $(1-\epsilon) \norm{x^TA}_2^2 - \lambda \norm{x}_2^2 \le \norm{x^T \tilde A}_2^2 \le (1+\epsilon) \norm{x^TA}_2^2 + \lambda \norm{x}_2^2$.}
\item Satisfies $\frac{\epsilon}{12} \cdot \frac{\norm{A-A_k}_F^2}{\norm{A-A_p}_F^2}$-Frobenius norm preservation (Definition \ref{def:frob}) for $\Ap$ where $p$ is the largest integer such that $\sigma_{p}^2 \ge \frac{\norm{A-A_k}_F^2}{k}$.
\end{enumerate}
Then $\tilde A = AS$ is an $(\epsilon,0,k)$-projection-cost-preserving sketch of $A$.
\end{theorem}
\begin{proof}
As in the proof of Theorem \ref{thm:matrixApprox}, letting $P \in \R^{n \times n}$ be any orthogonal projection matrix with rank at most $k$ and $Y  = I-P$, to prove the theorem it suffices to show \eqref{eq:initial}. Following Definition \ref{def:headTail} we decompose $A = A_p + \Ap$ where $p$ is the largest integer such that $\sigma_{p}^2 \ge \frac{\norm{A-A_k}_F^2}{k}$. Note that we always have $p \le 2k$ since $\norm{A-A_k}_F^2 = \sum_{i=k+1}^{\rank(A)} \sigma_{i}^2 \ge \sum_{i=k+1}^{2k} \sigma_{i}^2 \ge k \cdot \sigma_{2k}^2$. Using this decomposition, we can see that to prove the theorem it suffices to prove an analogous bound to \eqref{eq:initial3}:
\begin{align}
\underbrace{\left | \tr(Y A_p A_p^T Y)- \tr(Y A_p SS^T A_p^T Y) \right  |}_{\text{head term}}& + \underbrace{\left | \tr(Y \Ap \Ap^T Y) - \tr(Y\Ap SS^T \Ap^T Y) \right |}_{\text{tail term}}\nonumber\\
&+ \underbrace{2 \left |\tr(YA_p SS^T \Ap^T Y) \right |}_{\text{cross term}} \le \epsilon \norm{YA}_F^2.\label{eq:initial4}
\end{align}

We now proceed as in the proof of Theorem \ref{thm:matrixApprox}, bounding the three terms of \eqref{eq:initial4} separately.
\begin{claim}[Head Bound]\label{clm:head2}
Under the assumptions of Theorem \ref{thm:spectralApprox}, for any orthogonal projection matrix $P \in \R^{n \times n}$ with rank at most $k$ and $Y = I - P$:
\begin{align}\label{eq:head2}
| \tr(Y A_p A_p^T Y) - \tr(Y A_p SS^T A_p^T Y) | \le \frac{\epsilon}{12} \norm{YA}_F^2.
\end{align}
\end{claim}
\begin{proof}
Since for any $x \in \R^n$ we can write $x^T  A_p A_p^T x = y^T A A^T y$ for $y = U_pU_p^T x$, by our spectral error assumption we have:
\begin{align}
\label{eq:raw_top_bound}
\left (1-\frac{\epsilon}{24}  \right ){x}^T{A}_p{A}_p^T{x}   - \frac{\epsilon \|\Ak\|_F^2}{24k}\norm{y}_2^2
\leq{x}^T{A_p S}S^TA_p^T{x}
\leq \left (1+\frac{\epsilon}{24}  \right ){x}^T{A}_p{A}_p^T{x}  + \frac{\epsilon\|\Ak\|_F^2}{24k}\norm{y}_2^2.
\end{align}
By our choice of $p$, $y$ is orthogonal to all singular directions of ${A}$ except those with squared singular value greater than or equal to $\frac{\norm{\Ak}_F^2}{k}$. It follows that 
\begin{align*}
{x}^T{A}_p{A}_p^T{x} ={y}^T{A}{A}^T{y} \geq \frac{\|\Ak\|_F^2}{k}\cdot \norm{y}_2^2,
\end{align*}
and plugging back into \eqref{eq:raw_top_bound}, that  for any $x \in \R^n$:
\begin{align}\label{eq:headSpectral}
\left (1-\frac{\epsilon}{12}\right ) {x}^T{A}_p{A}_p^T{x} \le {x}^T{A_p S}S^TA_p^T{x} \le\left (1+\frac{\epsilon}{12}\right ) {x}^T{A}_p{A}_p^T{x}.
\end{align}
This yields \eqref{eq:head2}, completing the claim.
\end{proof}

\begin{claim}[Tail Bound]\label{clm:tail22}
Under the assumptions of Theorem \ref{thm:spectralApprox}, for any orthogonal projection matrix $P \in \R^{n \times n}$ with rank at most $k$ and $Y = I - P$:
\begin{align}\label{eq:tail22}
| \tr(Y \Ap \Ap^T Y) - \tr(Y \Ap SS^T \Ap^T Y) |  \le \frac{\epsilon}{6} \norm{YA}_F^2.
\end{align}
\end{claim}
\begin{proof}
As in the proof of Theorem \ref{thm:matrixApprox} (see equation \eqref{eq:tailSplit}) we bound the tail term as:
\begin{align}
| \tr(Y \Ap \Ap^T Y) - \tr(Y \Ap SS^T \Ap^T Y) | & \le \left | \norm{\Ap}_F^2 - \norm{\Ap S}_F^2 \right | + \left | \norm{P \Ap}_F^2 - \norm{P \Ap S}_F^2\right |.\label{eq:tailSplit2}
\end{align}
By the assumption that $S$ satisfies $\frac{\epsilon}{12} \cdot \frac{\norm{\Ak}_F^2}{\norm{\Ap}_F^2}$-Frobenius norm preservation (Definition \ref{def:frob}) for $\Ap$ we can bound the first term in \eqref{eq:tailSplit2} by:
\begin{align}\label{eq:tail12}
\left | \norm{\Ap}_F^2 - \norm{\Ap S}_F^2 \right |  \le \frac{\epsilon}{12} \norm{\Ak}_F^2.
\end{align}
We next bound the second term of \eqref{eq:tailSplit2}, $\left | \norm{P \Ap}_F^2 - \norm{P \Ap S}_F^2\right |$. For any $x \in \R^n$ we can write $x^T \Ap \Ap  x = y^T A A^T  y$ where $y  = (I-U_pU_p^T) x$. By  our spectral error assumption we have:
\begin{align*}
\left (1-\frac{\epsilon}{24}  \right ){x}^T \Ap \Ap^T{x}   - \frac{\epsilon \|\Ak\|_F^2}{24k}\norm{y}_2^2
\leq{x}^T{\Ap S}S^T\Ap^T{x}
\leq \left (1+\frac{\epsilon}{24}  \right ){x}^T\Ap\Ap^T{x}  + \frac{\epsilon\|\Ak\|_F^2}{24k}\norm{y}_2^2.
\end{align*}
Noting that $\norm{y}_2^2 \le \norm{x}_2^2$ and that by definition of $p$, $\norm{\Ap}_2^2 \le \frac{\norm{\Ak}_F^2}{k}$ and thus $x^T \Ap \Ap^T x \le \norm{x}_2^2 \cdot \frac{\norm{\Ak}_F^2}{k}$ we obtain:
\begin{align}\label{eq:totalAddBound}
\left |x^T (\Ap \Ap^T - \Ap SS^T \Ap^T)x \right | \le \frac{\epsilon \cdot \norm{x}_2^2 \cdot \norm{\Ak}_F^2}{12k}
\end{align}
Finally, since $P$ is a rank-$k$ projection matrix, we can write $P = ZZ^T$ where $Z \in \R^{n \times k}$ has orthonormal columns $z_1,...,z_k$: Using \eqref{eq:totalAddBound} we can bound:
\begin{align}\label{eq:secondTail2}
| \norm{ P \Ak }_F^2 - \norm{ P \Ak S }_F^2 | &= | \tr(P[\Ak \Ak^T - \Ak SS^T \Ak^T]P)| \nonumber \\ &= | \tr(ZZ^T[\Ak \Ak^T - \Ak SS^T \Ak^T]ZZ^T)| \nonumber \\ &=  | \tr(Z^T[\Ak \Ak^T - \Ak SS^T \Ak^T]Z)\nonumber\\
&\le \sum_{i=1}^k \frac{\epsilon \cdot \norm{z_i}_2^2 \cdot \norm{\Ak}_F^2}{12k}= \frac{\epsilon}{12} \cdot \norm{\Ak}_F^2,
\end{align}
where the second line follows from the cyclic property of trace and the fact that $Z^TZ = I$. Plugging \eqref{eq:tail12} and \eqref{eq:secondTail2} back into \eqref{eq:tailSplit2} gives:
\begin{align*}
| \tr(Y \Ap \Ap^T Y) - \tr(Y \Ap SS^T \Ap^T Y) | \le \frac{\epsilon}{12} \norm{YA}_F^2  + \frac{\epsilon}{12} \norm{\Ak}_F^2 \le \frac{\epsilon}{6} \norm{YA}_F^2,
\end{align*}
which gives \eqref{eq:tail22}, completing the claim.
\end{proof}

\begin{claim}[Cross Term Bound]\label{clm:cross2}
Under the assumptions of Theorem \ref{thm:spectralApprox}, for any orthogonal projection matrix $P \in \R^{n \times n}$ with rank at most $k$ and $Y = I - P$:
\begin{align}\label{eq:cross2s}
2 \left |\tr(YA_p SS^T \Ap^T Y) \right | \le \frac{\epsilon}{2} \norm{YA}_F^2.
\end{align}
\end{claim}
\begin{proof}
We follow equation \eqref{eq:crossSplit} in the proof of Claim \ref{clm:cross}, writing $C = AA^T$ and bounding
\begin{align}
2 \left |\tr(YA_p SS^T \Ap^T Y) \right | \le 2\sqrt{\tr(YCC^{+/2} C^{+/2}C Y)} \cdot \sqrt{\tr(\Ap SS^T A_p^T C^{+/2}C^{+/2} A_p SS^T \Ap^T)}.\label{eq:crossSplit2}
\end{align}
As in \eqref{eq:cross1} we have: 
\begin{align}
\sqrt{\tr(YCC^{+/2} C^{+/2}C Y)} = \norm{YA}_F\label{eq:cross12}.
\end{align}
It remains to bound the second term in the product. Following \eqref{eq:secondTermNorm} we have:
%
Let $\Sigma_p \in \R^{p \times p}$ be the top left $p \times p$ submatrix of $\Sigma$ (with diagonal entries $\sigma_1,...,\sigma_p$). We have $C^{+/2} = U \Sigma^{-1} U^T$ and recalling that $A_p = U_p U_p^T A$ can write: 
\begin{align}
\label{eq:sum_sep}
\tr(\Ap SS^T A_p^T C^{+/2}C^{+/2} A_p SS^T \Ap^T) &= \|\Ap SS^T A_p^T C^{+/2}\|_F^2\nonumber\\
&= \norm{\Ap SS^T A^T U_p U_p^T U \Sigma^{-1} U^T}_F^2 \nonumber\\
&= \norm{\Ap SS^T A^T U_p \Sigma_p^{-1}}_F^2 \nonumber\\
&= \sum_{i=1}^p \|\Ap SS^T A^T{u}_{i}\|_2^2\cdot \sigma_{i}^{-2}.
\end{align}
We prove that the summand is small for every $i$. Take $p_i$ to be a unit vector:
$$p_i \eqdef \frac{1}{\| \Ap SS^TA{u}_{i}\|_2} \cdot \Ap SS^TA^T u_i.$$
Note that $p_i$ falls within the column span of $\Ap$ and thus $p_i^T A = p_i^T \Ap$. Analogously, $u_i^T A = u_i^T A_p$. Using the first fact we can write:
\begin{align}\label{eq:squaredDot}
\|\Ap SS^T A^T{u}_{i}\|_2^2 = (p_i^T \Ap SS^TA^T u_i)^2 = (p_i^T A SS^TA^T u_i)^2.
\end{align}
Now, suppose we construct the vector ${m} = \left(\sigma_i^{-1}{u}_i + \frac{\sqrt{k}}{\|\Ak\|_F}{p}_i\right)$. By  our spectral error assumption we have that:
\begin{align*}
m^T A SS^T A^Tm &\leq \left (1+\frac{\epsilon}{24}\right ) m^T AA^T m + \frac{\epsilon \norm{\Ak}_F^2}{24k} \norm{m}_2^2,
\end{align*}
which expands to give:
\begin{align}
\label{cross_term_mid}
&\sigma_i^{-2}{u}_i^T ASS^TA^T{u}_i + \frac{k}{\|\Ak \|^2_F}{p}_i^TASS^TA{p}_i + \frac{2\sqrt{k}}{\sigma_i \|\Ak\|_F}{p}_i^TASS^T A^T u_i \nonumber \\
&\leq \left (1+\frac{\epsilon}{24} \right )\sigma_i^{-2}{u}_i^T{A}{A}^T{u}_i + \left (1+\frac{\epsilon}{24}\right )\frac{k}{\norm{\Ak}^2_F} p_i^T AA^T p_i  + \frac{\epsilon\|\Ak\|_F^2}{24k} \norm{m}_2^2.
\end{align}
There is no cross term on the right side since $p_i^TAA^T u_i = p_i^T \Ap A_p^T u_i = 0$.
From \eqref{eq:headSpectral}, and using that  $u_i^T A = u_i^T A_p$, we have:
\begin{align}\label{camTest1}
{u}_i^TASS^T A^T {u}_i \geq \left (1-\frac{\epsilon}{12}\right ){u}_i^TAA^T u_i = \left (1-\frac{\epsilon}{12}\right ) \sigma_i^2.
\end{align}
Additionally, from \eqref{eq:totalAddBound}, the fact that $p_i^T A = p_i^T \Ap$, and that $p_i$ is a unit vector:
\begin{align}\label{camTest2}
p_i^T ASS^T A p_i \ge p_i^T  AA^T p_i - \frac{\epsilon \norm{\Ak}_F^2}{12k}.
\end{align}
Plugging \eqref{camTest1} and \eqref{camTest2} back into \eqref{cross_term_mid} gives:
\begin{align}
\left (1-\frac{\epsilon}{12}\right ) &+ \frac{k}{\|\Ak\|^2_F}{p}_i^T{A}{A}^T{p}_i - \frac{\epsilon}{12} + \frac{2\sqrt{k}}{\sigma_i \|\Ak\|_F}{p}_i^T ASS^T A^T{u}_i \nonumber\\
&\leq \left (1+\frac{\epsilon}{24}\right)+ \left (1+\frac{\epsilon}{24}\right) \frac{k}{\|\Ak\|^2_F}{p}_i^T{A}{A}^T{p}_i + \frac{\epsilon\|\Ak\|_F^2}{24k} \norm{m}_2^2. \label{eq:cross_term_bound_2}
\end{align}
Noting that ${p}_i^T{A}{A}^T{p}_i \leq \frac{\|\Ak\|_F^2}{k}$ since ${p}_i$ lies in the column span of $\Ap$, rearranging \eqref{eq:cross_term_bound_2} gives:
\begin{align}\label{camtest3}
\frac{2\sqrt{k}}{\sigma_i \|\Ak\|_F}{p}_i^TASS^TA^T{u}_i \leq \frac{\epsilon}{4} + \frac{\epsilon\|\Ak \|_F^2}{24 k} \cdot \norm{m}_2^2 \leq \frac{\epsilon}{3}.
\end{align}
The second inequality above follows from the fact that for $i \le p$, $\sigma_i^{-2} \leq \frac{{k}}{\|\Ak\|_F^2}$ and that $u_i^T  p_i = 0$ so $\|{m}\|_2^2 = \sigma_i^{-2} \norm{u_i}_2^2 + \frac{k}{\norm{\Ak}_F^2} \norm{p_i} \leq \frac{2 k}{\|\Ak\|_F}.$ Squaring \eqref{camtest3} gives
\begin{align*}
({p}_i^TASS^TA^T{u}_i)^2 \leq \frac{\epsilon^2}{36}\cdot \frac{\sigma_i^2 \|\Ak\|^2_F}{k}.
\end{align*}
Plugging into \eqref{eq:sum_sep} using \eqref{eq:squaredDot} and that $p \le 2k$ then gives:
\begin{align}\label{crossSecondSpectral}
\tr(\Ap SS^T A_p^T C^{+/2}C^{+/2} A_p SS^T \Ap^T) \leq \sum_{i=1}^p \frac{\epsilon^2}{36}\cdot \frac{\sigma_i^2 \|\Ak\|^2_F}{k} \cdot \sigma_{i}^{-2} \leq \frac{\epsilon^2}{18} \|\Ak\|^2_F
\end{align}
Finally, plugging \eqref{eq:cross12} and \eqref{crossSecondSpectral} back into \eqref{eq:crossSplit2} gives:
\begin{align}
2 \left |\tr(YA_p SS^T \Ap^T Y) \right | &\leq 2 \norm{YA}_F \cdot \sqrt{\frac{1}{18}} \cdot \epsilon \norm{\Ak}_F \le \frac{\epsilon}{2} \norm{YA}_F^2,
\end{align}
which gives \eqref{eq:cross2s}, completing the claim.
\end{proof}

\medskip
\spara{Completing the Proof:}
\medskip

\noindent Finally, we combine the head, tail and cross term bounds of Claims \ref{clm:head2}, \ref{clm:tail22}, and \ref{clm:cross2} to give:
\begin{align}
{\left | \tr(Y A_p A_p^T Y)- \tr(Y A_p SS^T A_p^T Y) \right  |}& +{\left | \tr(Y \Ap \Ap^T Y) - \tr(Y\Ap SS^T \Ap^T Y) \right |} +{2 \left |\tr(YA_p SS^T \Ap^T Y) \right |}\nonumber \\
&\le \frac{\epsilon}{12} \norm{YA}_F^2 + \frac{\epsilon}{6} \norm{YA}_F^2 + \frac{\epsilon}{2} \norm{YA}_F^2 \le \epsilon  \norm{YA}_F^2.
\end{align}
This yields \eqref{eq:initial3}, completing the theorem.
\end{proof}

\subsection{Constructions Satisfying Theorem \ref{thm:spectralApprox}}
The spectral approximation and Frobenius norm preservation requirements of Theorem \ref{thm:spectralApprox} are satisfied by many sketching methods. The are particularly natural in proving projection-cost-preserving sketch properties for column selection methods, two of which we use as examples below.

\begin{corollary}[Ridge Leverage Score Sampling]\label{cor:ridge}
Let $a_i \in \R^n$ be the $i^{th}$ column of $A$. The $i^{th}$ $\lambda$-ridge leverage score of $A$ is given by $$\tau_i(A) \eqdef {a}_i^T (AA^T + \lambda I)^{-1} a_i.$$
For every $i$, let $\tilde \tau_i \ge \tau_i(A)$ be an overestimate for the $i^{th}$ $\lambda$-ridge leverage score with $\lambda = \frac{\norm{A-A_k}_F^2}{k}$. Let $p_i = \frac{\tilde \tau_i}{\sum_{i=1}^d \tilde \tau_i}$ and let $t = \frac{c \log(k/\delta)}{\epsilon^2} \sum_{i=1}^d \tilde \tau_i$ for any $\epsilon,\delta \in (0,1)$ and some sufficiently large constant $c$. Let $S  \in \R^{d \times t}$ be a sampling matrix selecting $t$ columns of $A$, where each column of $S$ is set independently to $\frac{1}{\sqrt{tp_i}} e_i$ with probability $p_i$, where $e_i$ is the $i^{th}$ standard basis vector. Then, with probability $\ge 1-\delta$, $S$ satisfies the conditions of Theorem \ref{thm:spectralApprox} and hence $\tilde A = AS$ is an $(\epsilon,0,k)$-projection-cost-preserving sketch of $A$.
\end{corollary}
Note that $\sum_{i=1}^d \tau_i(A) \le 2k$ (see e.g., Lemma 4 of \cite{cohen2017input}) and thus if the approximate ridge leverage scores are within a constant factor of the true ones, $\sum_{i=1}^d \tilde \tau_i = O(k)$ and so $\tilde A$ has $O(k\log(k/\delta)/\epsilon^2)$ columns.
\begin{proof}
The spectral approximation guarantee can be proven with a matrix Bernstein inequality. See Theorem 5 of \cite{cohen2017input}. The Frobenius norm preservation guarantee can be proven with a standard scalar Chernoff bound. See Lemma 20 of \cite{cohen2017input}.
\end{proof}

%

\begin{corollary}[Deterministic Column Selection]\label{cor:det}
There is a deterministic poly-time algorithm that, given $A \in \R^{n \times d}$, $\epsilon \in (0,1)$ returns sampling matrix $S \in \R^{d \times O(k/\epsilon^2)}$ satisfying the conditions of Theorem \ref{thm:spectralApprox} and thus that $\tilde A = AS$ is an $(\epsilon,0,k)$-projection-cost-preserving sketch of $A$.\footnote{Each column of $S$ is a scaled standard basis vector so $\tilde A = AS$ consists of a subset of reweighted columns of $A$.}
\end{corollary}
\begin{proof}
This corollary  follows easily from a stable-rank approximate matrix multiplication result given in \cite{cohen2015optimal}:
\begin{theorem}[Theorem 5 of \cite{cohen2015optimal}]\label{thm:cohen}
For any $k > 0$ and $\epsilon \in (0,1)$
there is a deterministic polynomial-time algorithm that, given $B \in \R^{n \times d}$ with $\norm{B}_2^2 \le 1$ and $\norm{B}_F^2 \le k$ returns sampling matrix $S \in \R^{d \times O(k/\epsilon^2)}$ satisfying
\begin{align}
\norm{BSS^T B^T - BB^T}_2 \le \epsilon.
\end{align}
\end{theorem}
For $\lambda = \frac{\norm{A-A_k}_F^2}{k}$ we set $B_1 = (AA^T + \lambda I)^{-1/2} A$ and let $b_2 \in \R^{d}$ be the vector whose $i^{th}$ entry  is equal to $\frac{\norm{(\Ap)_i}_2}{\norm{\Ap}_F}$, where $p$ is as defined in Theorem \ref{thm:spectralApprox} and $(\Ap)_i$ is the $i^{th}$ column of $\Ap$. Let $B = \frac{1}{2} [B_1;b_2]$ (that is, $B \in \R^{{n+1} \times d}$ is $\frac{1}{2} B_1$  with $\frac{1}{2} b_2$ appended as a final row.) Note that $B$ can be computed in polynomial time. Additionally we have:
\begin{align}\label{eq:bsBound}
\norm{B}_2^2 \le\frac{1}{4} \left ( 2\norm{B_1}_2^2 + 2\norm{b_2}_2^2 \right ) \le \frac{1}{4} (2 + 2) = 1
\end{align}
and 
\begin{align}\label{eq:bfBound}
\norm{B}_F^2 = \frac{1}{4} \left ( \norm{B_1}_F^2 + \norm{b_2}_2^2 \right ) \le \frac{1}{4} (2k + 1) \le k
\end{align}
where the second to last inequality follows since 
$$\norm{B}_F^2  = \sum_{i=1}^{\rank(A)} \sigma_i^2(B) = \sum_{i=1}^{\rank(A)} \frac{\sigma_i^2(A)}{\sigma_i^2(A) + \lambda} \le \sum_{i=1}^k \frac{\sigma_i^2(A)}{\sigma_i^2(A)} + \sum_{i=k+1}^{\rank(A)} \frac{\sigma_i^2(A)}{\norm{A-A_k}_F^2/k} = 2k.$$
\eqref{eq:bsBound} and \eqref{eq:bfBound} allow us to apply Theorem \ref{thm:cohen} to $B$ with error parameter $\epsilon/48$ obtaining $S$ with $O(k/\epsilon^2)$ rows satisfying:
\begin{align}\label{eq:b96}
\norm{BSS^T B^T - BB^T}_2 \le \frac{\epsilon}{96}.
\end{align}
\eqref{eq:b96} implies first that 
$$\frac{1}{4} \left | b_2^T S S^T b_2  - b_2^T  b_2 \right | = \frac{1}{4\norm{\Ap}_F^2} \cdot \left | \norm{\Ap S}_F^2  - \norm{\Ap}_F^2 \right |  \le \frac{\epsilon}{96},$$
which gives that 
\begin{align}\label{eq:bfFrobReduce}
\left |  \norm{\Ap S}_F^2  - \norm{\Ap}_F^2 \right | \le \frac{\epsilon}{24} \norm{\Ap}_F^2 \le \frac{\epsilon}{12} \norm{\Ak}_F^2,
\end{align} 
where the last inequality follows from that fact that $\norm{\Ap}_F^2 \le 2 \norm{\Ak}_F^2$. If $p \ge k$ this is true immediately. Otherwise, if $p < k$ it follows since:
\begin{align*}
\norm{\Ap}_F^2 - \norm{\Ak}_F^2 = \sum_{i=p+1}^{k} \sigma_i^2 \le k \cdot \sigma_{p+1}^2 \le \norm{\Ak}_F^2
\end{align*}
and so $\norm{\Ap}_F^2 \le 2\norm{\Ak}_F^2$. \eqref{eq:bfFrobReduce} gives the Frobenius norm preservation condition of Theorem \ref{thm:spectralApprox}. From \eqref{eq:b96} we can also conclude that 
$$B_1 B_1 ^T -\frac{\epsilon}{24} I \preceq B_1SS^T B_1^T \preceq B_1 B_1 ^T +\frac{\epsilon}{24} I $$ which after multiplying by $(AA^T + \lambda I)^{1/2}$ on the right and left (recalling that $B_1 = (AA^T + \lambda I)^{-1/2} A$) gives:
\begin{align*}
AA^T - \frac{\epsilon}{24} (AA^T + \lambda I) \preceq ASS^T A^T \preceq AA^T + \frac{\epsilon}{24} (AA^T + \lambda I),
\end{align*}
which gives the spectral approximation condition of Theorem \ref{thm:spectralApprox}, completing the proof.
\end{proof}
We note that \cite{cohen2015optimal} proves that a number of sketching methods satisfy the stable-rank approximation matrix multiplication result of Theorem \ref{thm:cohen}. We can use an analogous proof to that of Corollary  \ref{cor:det} to prove that all these methods yield projection-cost-preserving sketches via Theorem \ref{thm:spectralApprox}.

\bibliographystyle{alpha}
\bibliography{pcp}	

\newcommand{\etalchar}[1]{$^{#1}$}
\begin{thebibliography}{CEM{\etalchar{+}}15}

\bibitem[CEM{\etalchar{+}}15]{cohen2015dimensionality}
Michael~B. Cohen, Sam Elder, Cameron Musco, Christopher Musco, and Madalina
  Persu.
\newblock Dimensionality reduction for k-means clustering and low rank
  approximation.
\newblock In {\em \STOC{2015}}, pages 163--172, 2015.
\newblock https://arxiv.org/abs/1410.6801.

\bibitem[CMM17]{cohen2017input}
Michael~B. Cohen, Cameron Musco, and Christopher Musco.
\newblock Input sparsity time low-rank approximation via ridge leverage score
  sampling.
\newblock In {\em \SODA{2017}}, pages 1758--1777, 2017.

\bibitem[CNW16]{cohen2015optimal}
Michael~B. Cohen, Jelani Nelson, and David~P. Woodruff.
\newblock Optimal approximate matrix product in terms of stable rank.
\newblock In {\em \ICALP{2016}}, 2016.
\newblock https://arxiv.org/abs/1507.02268.

\bibitem[CYD19]{chowdhury2019structural}
Agniva Chowdhury, Jiasen Yang, and Petros Drineas.
\newblock Structural conditions for projection-cost preservation via randomized
  matrix multiplication.
\newblock {\em Linear Algebra and its Applications}, 2019.

\bibitem[FSS13]{feldman2013turning}
Dan Feldman, Melanie Schmidt, and Christian Sohler.
\newblock Turning big data into tiny data: Constant-size coresets for k-means,
  {PCA} and projective clustering.
\newblock In {\em \SODA{2013}}, pages 1434--1453, 2013.

\bibitem[KN14]{KaneNelson:2014}
Daniel~M. Kane and Jelani Nelson.
\newblock Sparser {J}ohnson-{L}indenstrauss transforms.
\newblock {\em J. ACM}, 61(1), January 2014.

\bibitem[Mus15]{masters}
Cameron Musco.
\newblock Dimensionality reduction for k-means clustering.
\newblock Master's thesis, MITs, 2015.
\newblock https://dspace.mit.edu/handle/1721.1/101473.

\bibitem[MW17]{musco2017sublinear}
Cameron Musco and David~P. Woodruff.
\newblock Sublinear time low-rank approximation of positive semidefinite
  matrices.
\newblock In {\em \FOCS{2017}}, pages 672--683, 2017.
\newblock https://arxiv.org/abs/1704.03371.

\bibitem[Sar06]{Sarlos:2006}
Tamas Sarlos.
\newblock Improved approximation algorithms for large matrices via random
  projections.
\newblock In {\em \FOCS{2006}}, pages 143--152, 2006.

\bibitem[Ver18]{Vershynin:2018}
Roman Vershynin.
\newblock {\em High-Dimensional Probability: An Introduction with Applications
  in Data Science}.
\newblock Cambridge Series in Statistical and Probabilistic Mathematics.
  Cambridge University Press, 2018.

\bibitem[Wai19]{Wainwright:2019}
Martin~J. Wainwright.
\newblock {\em High-Dimensional Statistics: A Non-Asymptotic Viewpoint}.
\newblock Cambridge Series in Statistical and Probabilistic Mathematics.
  Cambridge University Press, 2019.

\bibitem[Woo14]{woodruff2014sketching}
David~P. Woodruff.
\newblock Sketching as a tool for numerical linear algebra.
\newblock {\em Foundations and Trends in Theoretical Computer Science},
  10(1--2):1--157, 2014.

\end{thebibliography}

\end{document}